%% file: Arxiv.tex
\documentclass[11pt,journal,onecolumn]{IEEEtran}

\usepackage[margin=2.2cm]{geometry}

\usepackage{url}            
\usepackage{booktabs}       
\usepackage{amsfonts}       
\usepackage{nicefrac}       
\usepackage{microtype}      
\usepackage{caption}

\usepackage{array}
\usepackage{makecell}

\input{math_commands}

\usepackage{bbm}
\usepackage[sort, numbers]{natbib}

\usepackage[utf8]{inputenc} 
\usepackage[T1]{fontenc}    
\usepackage{hyperref}
\hypersetup{colorlinks=true,linkcolor=red,citecolor = blue, linktocpage}

\usepackage{algorithmic}
\usepackage{graphicx}
\usepackage{textcomp}
\usepackage{epstopdf}
\usepackage{times}
\usepackage{amssymb}
\usepackage{amsthm}

\usepackage{xcolor}

\makeatletter
\newtheorem*{rep@theorem}{\rep@title}
\newcommand{\newreptheorem}[2]{%
\newenvironment{rep#1}[1]{%
 \def\rep@title{#2 \ref{##1}}%
 \begin{rep@theorem}}%
 {\end{rep@theorem}}}
\makeatother

\newtheorem{theorem}{Theorem}
\newreptheorem{theorem}{Theorem}

\newreptheorem{lemma}{Lemma}

\newtheorem{proposition}{Proposition}
\newreptheorem{proposition}{Proposition}

\newtheorem{corollary}{Corollary}

\newtheorem{remark}{Remark}
\usepackage{url}            
\usepackage{booktabs}       
\usepackage{amsfonts}       
\usepackage{nicefrac}       
\usepackage{microtype}      

 \usepackage{cite}
\usepackage{algorithmic}
\usepackage{graphicx}
\usepackage{textcomp}
\usepackage{epstopdf}
\usepackage{times}
\usepackage{amssymb}
\usepackage{amsthm}
\usepackage{graphics}
\usepackage{mathrsfs}
\usepackage{psfrag}
\usepackage{tikz}

\usepackage{xcolor}
\defcitealias{OurGibbsPaper}{Aminian${}^{\star}$ and Bu${}^{\star}$ et al., 2021}
\newcommand{\nn}{\nonumber}

\begin{document}

%

%


\title{Tighter Expected Generalization Error Bounds via Convexity of Information Measures}

\author{
        Gholamali Aminian$^*$,
        Yuheng Bu$^*$,
        Gregory Wornell,
        Miguel Rodrigues
\thanks{$^*$ Equal Contribution.}
\thanks{G. Aminian and M. Rodrigues are with the Electronic and Electrical Engineering Department at University College London, UK, (Email: g.aminian, m.rodrigues@ucl.ac.uk).}
\thanks{Y. Bu and G. Wornell are with the Department of Electrical Engineering and Computer Science, Massachusetts Institute of Technology, Cambridge, MA 02139 (Email: buyuheng, gww@mit.edu).}

}

\allowdisplaybreaks

\maketitle


\begin{abstract}

Generalization error bounds are essential to understanding machine learning algorithms. This paper presents novel expected generalization error upper bounds based on the average joint distribution between the output hypothesis and each input training sample. Multiple generalization error upper bounds based on different information measures are provided, including Wasserstein distance, total variation distance, KL divergence, and Jensen-Shannon divergence. Due to the convexity of the information measures, the proposed bounds in terms of Wasserstein distance and total variation distance are shown to be tighter than their counterparts based on individual samples in the literature. An example is provided to demonstrate the tightness of the proposed generalization error bounds.
\end{abstract}

\section{Introduction}
Machine learning algorithms are increasingly adopted to solve various problems in a wide range of applications. Understanding the generalization behavior of a learning algorithm is one of the most important challenges in statistical learning theory. Various approaches have been developed to bound the generalization  error~\citep{rodrigues2021information}, including VC dimension-based bounds~\citep{vapnik1999overview}, algorithmic stability-based bounds ~\citep{bousquet2002stability}, algorithmic robustness-based bounds ~\citep{xu2012robustness}, PAC-Bayesian bounds~\citep{mcallester2003pac}.

More recently, approaches leveraging information-theoretic tools have been developed to characterize the generalization error of a learning algorithm. Such approaches incorporate various ingredients associated with a supervised learning problem, including the data generating distribution, the hypothesis space, and the learning algorithm itself, expressing expected generalization error in terms of specific information measures between the input of training dataset and output hypothesis.

In particular, building upon pioneering work by Russo and Zou~\citep{russo2019much}, an expected generalization error upper based on the mutual information between the training set and the hypothesis is proposed by Xu and Raginsky~\citep{xu2017information}. Bu \textit{et al.}~\citep{bu2020tightening} have derived tighter generalization error bounds based on individual sample mutual information. The generalization error bounds based on other information measures such as $\alpha$-R\'eyni divergence~\citep{modak2021renyi}, maximal leakage~\citep{esposito2019generalization}, Jensen-Shannon divergence~\citep{aminian2020jensen}, Wasserstein distances~\citep{lopez2018generalization,wang2019information} and individual sample Wasserstein distance~\citep{galvez2021tighter} are also considered. Chaining mutual information technique is proposed in~\citep{asadi2018chaining} and~\citep{asadi2020chaining} to further improve the mutual information-based bound. The upper bounds based on conditional mutual information and individual sample conditional mutual information are proposed in \citep{steinke2020reasoning} and \citep{zhou2020individually}, respectively. It is shown in ~\citep{hafez2020conditioning,haghifam2020sharpened}, the combination of conditioning and processing techniques could provide tighter generalization error upper bounds. Using rate-distortion theory, \citep{masiha2021learning,bu2020information,bu2021population} provide information-theoretic generalization error upper bounds for model misspecification and model compression, respectively. An exact characterization of the generalization error for the Gibbs algorithm in terms of symmetrized KL information is provided in \citep{Our2021exact}.

In this paper, we introduce the notion of \emph{average joint distribution}, which is the average of the distribution between the output hypothesis and each training sample. We aspire to provide a more refined analysis of the generalization ability of randomized learning algorithms by representing the expected generalization error using the aforementioned average joint distribution. The merit of this representation is that it directly leads to some tighter generalization error upper bounds based on the convexity of the information measures, including Wasserstein distance and total variation distance. The proposed bound finds its application when the importance of each training sample is not the same in the learning algorithm, e.g.,  imbalanced classification or learning under noisy data samples. 

More specifically, our contributions are as follows:
\begin{itemize}
    \item We provide novel expected generalization error upper bounds based on the average joint distribution between the output hypothesis and each input training sample, in terms of Wasserstein distance, total variation distance, KL divergence, and Jensen-Shannon divergence.
    \item  We offer an upper bound on the difference between the empirical risk of two learning algorithms using the KL divergence between average joint distributions.
    \item We construct a simple numerical example to demonstrate the improvement of the proposed upper bound based on the average joint distribution in comparison to individual sample mutual information bound~\citep{bu2020tightening}.
\end{itemize}

\textbf{Notations:} 
A random variable is denoted by an upper-case letter (e.g., $Z$), its alphabet is denoted by the corresponding calligraphic letter (e.g., $\mathcal{Z}$), and the realization of the random variable is denoted with a lower-case letter (e.g., $z$). 
The probability distribution of the random variable $Z$ is denoted by $P_Z$. The joint distribution of a pair of random variables $(Z_1,Z_2)$ is denoted by $P_{Z_1,Z_2}$. 

\textbf{Information Measures:} The differential entropy of a continuous probability measure $P$ defined over space $\mathcal{Z}$ is given by $h(P)\triangleq\int_\mathcal{Z} -dP \log(dP)$.
If  $P$ and $Q$ are probability measures defined over space $\mathcal{Z}$, and $P$ is absolutely continuous with respect to $Q$, the Kullback-Leibler (KL) divergence between $P$ and $Q$ is given by
$D(P\|Q)\triangleq\int_\mathcal{Z}\log\left(\frac{dP}{dQ}\right) dP$. The Donsker-Varadhan variational representation of the KL divergence is as follows~\citep{polyanskiy2014lecture},
\begin{align}\label{eq: KL rep}
    D(P\|Q)=\sup_{g\in \mathcal{G}}\left\{\mathbb{E}_P[g(Z)]-\log(\mathbb{E}_Q[e^{g(Z)}])\right\},
\end{align}
where the supremum is over all measurable functions, i.e., $\mathcal{G}=\{g: \mathcal{Z}\to \mathbb{R}, \text{ s.t. } \mathbb{E}_Q[e^{g(Z)}]<\infty\}$. 

The Jensen-Shannon divergence \citep{lin1991divergence} is defined as 
\begin{equation}
    D_{JS}(P\|Q)\triangleq\frac{D(P\| \frac{P+Q}{2})}{2}+\frac{D(Q\| \frac{P+Q}{2})}{2}, 
\end{equation}
and it can be verified that $D_{JS}(P\|Q)\leq \log(2)$.

The mutual information between two random variables $X$ and $Y$ is defined as the KL divergence between their joint distribution and the product of the marginals, i.e., $I(X;Y)\triangleq D(P_{X,Y}\|P_X\otimes P_{Y})$. Similarly, the Lautum information introduced in  \citep{palomar2008lautum} is defined as the KL divergence between the product of the marginals and the joint distribution, i.e.,  $L(X;Y)\triangleq D(P_X\otimes P_{Y} \| P_{X,Y})$.

The Wasserstein distance between $P$ and $Q$ is defined using a metric $\rho : \mathcal{Z}\times\mathcal{Z} \to \mathbb{R}_0^+$, and it is given by:
\begin{equation}
    \mathbb{W}(P,Q)=\inf_{\pi \in \Pi(P,Q)}\int_{\mathcal{Z} \times \mathcal{Z}} \rho(z,z') d\pi(z,z'),
\end{equation}
where $\Pi(P,Q)$ is the set of all joint distributions $\pi$ over the product space $\mathcal{Z} \times \mathcal{Z}$ with marginal distributions $P$ and $Q$.
When $\mathcal{Z}$ is a normed space with norm $\|\cdot\|$, simply taking $\rho(z,z') = \| z - z' \|$ leads to
\begin{equation}
    \mathbb{W}(P,Q)\triangleq\inf_{\pi \in \Pi(P,Q)}\int_{\mathcal{Z} \times \mathcal{Z}} \| z-z' \| \ d\pi(z,z').
\end{equation}

Another representation for the Wasserstein distance is given by the Kantorovich-Rubinstein duality~\citep{villani2009optimal}, i.e.,
\begin{equation}\label{Eq: KR rep}
    \mathbb{W}(P,Q)
    =\sup_{g\in \{g:\operatorname{Lip}(g) \leq 1\}} \left\{ \mathbb{E}_{P}[g(Z)]-\mathbb{E}_{Q}[g(Z)] \right\},
\end{equation}
where $\operatorname{Lip}(g)$ denotes the Lipschitz constant of function $g : \mathcal{Z} \to \mathbb{R}$, namely
\begin{align}
    \operatorname{Lip}(g) 
     \triangleq &\inf \big\{ L>0 :  |g(z_1)-g(z_2)| \leq L \| z_1 - z_2 \|,\  z_1,z_2 \in \mathcal{Z} \big\}. \nn
\end{align}

The total variation distance between $P$ and $Q$ is given by \begin{equation}\label{Eq: tv def}
    \mathbb{TV}(P,Q)\triangleq  \frac{1}{2}\int |dP-dQ|.
\end{equation}
Note that total variation distance also arises from Wasserstein distance~\citep{villani2009optimal}, i.e., $\mathbb{TV}(P,Q) = \mathbb{W}(P,Q)$, when $\rho(z,z')=\mathbbm{1}\{z\neq z'\}$ where $\mathbbm{1}$ is an indicator function.

\section{Problem Formulation}

Let $S = \{Z_i\}_{i=1}^n$ be the training set, where each $Z_i$ is defined on the same alphabet $\mathcal{Z}$. Note that $Z_i$ is not required to be i.i.d generated from the same data-generating distribution $P_Z$, and we denote the joint distribution of all the training samples as $P_S$.
We denote the hypotheses by $w \in \mathcal{W}$, where $\mathcal{W}$ is a hypothesis class. The performance of the hypothesis is measured by a non-negative loss function $\ell:\mathcal{W} \times \mathcal{Z}  \to \mathbb{R}_0^+$, and we can define the empirical risk and the population risk associated with a given hypothesis $w$ as 
\begin{align}
    &L_E(w,s)\triangleq\frac{1}{n}\sum_{i=1}^n \ell(w,z_i), \\
    &    L_P(w,P_S)\triangleq  \mathbb{E}_{P_S}[L_E(w,S)],
\end{align}
respectively. A learning algorithm can be modeled as a randomized mapping from the training set $S$ onto an hypothesis $W\in\mathcal{W}$  according to the conditional distribution $P_{W|S}$. Thus, the expected generalization error quantifying the degree of over-fitting can be written as
\begin{equation}\label{Eq: expected GE}
\overline{\text{gen}}(P_{W|S},P_S)\triangleq\mathbb{E}_{P_{W,S}}[ L_P(W,P_S)-L_E(W,S)],
\end{equation}
where the expectation is taken over the joint distribution $P_{W,S} =  P_{W|S}\otimes P_S$.

In this paper, we construct different upper bounds for generalization error using  the average joint distribution, which is defined as 
\begin{align}\label{Eq: average joint distribution}
    \overline{P}_{W,\overline{Z}}(w,z) \triangleq \frac{1}{n}\sum_{i=1}^n P_{W,Z_i}(w,z).
\end{align}
Note that the average sample distribution is defined as
\begin{align}
   \overline{P}_{\overline{Z}}(z) \triangleq \frac{1}{n}\sum_{i=1}^n P_{Z_i}(z).
\end{align}
It is worthwhile to mention that under i.i.d assumption we have $\overline{P}_{\overline{Z}}=P_Z$.
Similarly, the average conditional distribution is defined as
\begin{equation}\label{Eq: sym learning algorithm}
    \overline{P}_{W|\overline{Z}=z}(w)=\frac{1}{n}\sum_{i=1}^n P_{W|Z_i=z}(w).
\end{equation}
A learning algorithm is said to be \textit{symmetric}, if the conditional distributions between each sample $Z_i$ and hypothesis $W$ are the same, i.e., 
$P_{W|Z_i}=P_{W|Z},\  \forall i\in\{1,\cdots,n\}$.


\section{Generalization Error Upper Bounds}
This section provides expected generalization error upper bounds in terms of different information measures, including Wasserstein distance, total variation distance, KL divergence, and Jensen-Shannon divergence. In the case of Wasserstein distance and total variation distance, our upper bounds are shown to be tighter than existing upper bounds based on these information measures.

To present our result, we first show that the expected generalization error can be expressed in terms of the average joint distribution~\eqref{Eq: average joint distribution}.
\begin{proposition}\label{Prop: new exp gen}
The expected generalization error of a learning algorithm $P_{W|S}$ can be written as
\begin{equation}\label{Eq: new  rep expected GE}
\overline{\text{gen}}(P_{W|S},P_S)=\mathbb{E}_{P_W \otimes  \overline{P}_{\overline{Z}}}[\ell(W,Z)]-\mathbb{E}_{ \overline{P}_{W,\overline{Z}}}[\ell(W,Z)].
\end{equation}
\end{proposition}
\begin{proof}
By the definition of generalization error, we have
\begin{align}\nn
&\overline{\text{gen}}(P_{W|S},P_S)=\mathbb{E}_{P_{W,S}}[ L_P(W,P_S)-L_E(W,S)]\\
&= \frac{1}{n}\sum_{i=1}^n \mathbb{E}_{P_{W}\otimes P_{Z_i}}[\ell(W,Z)]-\frac{1}{n}\sum_{i=1}^n \mathbb{E}_{P_{W,Z_i}}[\ell(W,Z)]\\\label{Eq: linear exp}
&= \mathbb{E}_{P_{W}\otimes \frac{1}{n}\sum_{i=1}^n P_{Z_i}}[\ell(W,Z)]- \mathbb{E}_{\frac{1}{n}\sum_{i=1}^n P_{W,Z_i}}[\ell(W,Z)],\nn
\end{align}
where the last line follows by the linearity of expectation.
\end{proof}
This characterization embodied in Proposition~\ref{Prop: new exp gen} leads directly to various generalization error bounds in terms of different information measures.

\subsection{Wasserstein Distance-based Upper Bound}
In the following theorem, we provide a generalization error upper bound based on Wasserstein distance using~\eqref{Eq: new  rep expected GE} under Lipschitz condition.
\begin{theorem}\label{Theorem: Wasserine distance}
Suppose that for all $z\in \mathcal{Z}$, the loss function $\ell(\cdot,z)$ is L-Lipschitz, and we have i.i.d. training samples $S = \{Z_i\}_{i=1}^n$. Then, we have the following upper bound
\begin{align}
    |\overline{\text{gen}}(P_{W|S},P_S)|\leq L \mathbb{E}_{P_{{Z}}}[\mathbb{W}(\overline{P}_{W|\overline{Z}},P_{W})].
\end{align}
\end{theorem}
\begin{proof}
We have $\overline{P}_{\overline{Z}}=P_Z$ for i.i.d. data samples, then $\overline{P}_{W,\overline{Z}} = \overline{P}_{W|\overline{Z}} \otimes P_Z$. By ~\eqref{Eq: new  rep expected GE}, we have
\begin{align}\label{Eq: cond expected GE}\nn
|\overline{\text{gen}}(P_{W|S},P_S)|&=|\mathbb{E}_{P_Z}[\mathbb{E}_{P_W}[\ell(W,Z)]-\mathbb{E}_{ \overline{P}_{W|\overline{Z}}}[\ell(W,Z)]]|\\
&\leq L\mathbb{E}_{P_Z}[\mathbb{W}(\overline{P}_{W|\overline{Z}},P_{W})] ,
\end{align}
where the last inequality follows from Kantorovich-Rubinstein duality~\eqref{Eq: KR rep}.
\end{proof}
In the following, we show that our upper bound in Theorem~\ref{Theorem: Wasserine distance} would be tighter than the individual sample Wasserstein distance upper bound in~\citep{galvez2021tighter}.

\begin{proposition}\label{Prop: wassertine based on joint dis}
Under the same assumption as in Theorem~\ref{Theorem: Wasserine distance}, the upper bound in Theorem~\ref{Theorem: Wasserine distance} is always no worse  than the upper bound in \citep[Theorem~1]{galvez2021tighter}, i.e.,
\begin{align}
    |\overline{\text{gen}}(P_{W|S},P_S)|&\leq L \mathbb{E}_{P_Z}[\mathbb{W}(\overline{P}_{W|\overline{Z}},P_{W})] \nn \\
    &\leq \frac{L}{n} \sum_{i=1}^n\mathbb{E}_{P_Z}[\mathbb{W}(P_{W|Z_i},P_{W})].
\end{align}
\end{proposition}
\begin{proof}
By Kantorovich-Rubinstein duality~\eqref{Eq: KR rep}, we have
\begin{align}
    L\mathbb{W}(\overline{P}_{W|\overline{Z}},P_{W})&=\sup_{g\in \{g:\operatorname{Lip}(g) \leq 1\}}\{\mathbb{E}_{\overline{P}_{W|\overline{Z}}}[g]-\mathbb{E}_{P_{W}}[g] \}\nn\\
    &\leq \frac{1}{n}\sum_{i=1}^n \sup_{g\in \{g:\operatorname{Lip}(g) \leq 1\}}\{\mathbb{E}_{P_{W|Z_i}}[g]-\mathbb{E}_{P_{W}}[g] \}\nn\\
    &=\frac{1}{n} \sum_{i=1}^n\mathbb{W}(P_{W|Z_i},P_{W}),
\end{align}
where the inequality follows from convexity of supremum function.
\end{proof}
\begin{remark}\label{remark}
The upper bound based on average conditional distribution in Theorem~\ref{Theorem: Wasserine distance} will reduce to the individual sample Wasserstein distance based upper bound in \citep[Theorem~1]{galvez2021tighter} when the learning algorithm $P_{W|S}$ is symmetric, i.e., $P_{W|Z_i}=P_{W|Z}$ for all $i$.  
\end{remark}

\subsection{Total Variation Distance-based Upper Bound}
In the following result, we provide a tighter expected generalization error upper bound in terms of total variation distance for bounded loss functions.
\begin{proposition}\label{Prop: total variation}
Suppose that the loss function is bounded, i.e., $\ell \in [a,b]$, and we have i.i.d. training samples $S = \{Z_i\}_{i=1}^n$. Then, the following upper bound holds
 \begin{align}\label{Eq: upper based tv}
     |\overline{\text{gen}}(P_{W|S},P_S)|&\leq (b-a) \mathbb{E}_{P_{Z}}[\mathbb{TV}(\overline{P}_{W|\overline{Z}},P_W)]\nn\\
     &=\mathbb{TV}(\overline{P}_{W,\overline{Z}},P_W \otimes P_Z).
 \end{align}
\end{proposition}
\begin{proof}
The bounded condition implies that the loss function $\ell(\cdot,z)$ is $(b-a)$-Lipschitz for all $z\in \mathcal{Z}$. Recall that total variation distance is a special case of Wasserstein distance with $\rho(z,z')=\mathbbm{1}\{z\neq z'\}$, then the inequality can be proved by applying Theorem~\ref{Theorem: Wasserine distance} directly. 

By the assumption of i.i.d. training samples and the definition of total variation in~\eqref{Eq: tv def}, we have
\begin{equation}
    \mathbb{E}_{P_{Z}}[\mathbb{TV}(\overline{P}_{W|\overline{Z}},P_W)]=\mathbb{TV}(\overline{P}_{W,\overline{Z}},P_W \otimes P_Z),
\end{equation}
which completes the proof for the equality.
\end{proof}

Next, we compare our upper bound in terms of total variation distance with the individual sample total variation distance based upper bound in \citep[Corollary 1]{galvez2021tighter}.
\begin{corollary}\label{Cor: compare TV}
Under the same assumptions as in Proposition~\ref{Prop: total variation}, the  upper bound in Proposition~\ref{Prop: total variation}
is always no worse than the individual sample total variation distance bound in \citep[Corollary 1]{galvez2021tighter}, i.e., 
\begin{align}\label{Eq: TV1}
     |\overline{\text{gen}}(P_{W|S},P_S)|&\leq (b-a) \mathbb{E}_{P_Z}[\mathbb{TV}(\overline{P}_{W|\overline{Z}},P_W)] \nn \\
     &\leq \frac{(b-a)}{n} \sum_{i=1}^n\mathbb{E}_{P_Z}[\mathbb{TV}(P_{W|Z_i},P_W)].
 \end{align}
\end{corollary}
\begin{proof}
As the total variation is an $f$-divergence, it has the joint convexity property with respect to its input~\citep{polyanskiy2014lecture}. Thus, the result follows by applying the convexity of the total variation distance in \eqref{Eq: TV1}.
\end{proof}

\begin{remark}
Under the same assumptions as in Proposition~\ref{Prop: total variation}, it is shown in~\citep[Corollary~1]{galvez2021tighter} that the upper bound based on individual sample total variation distance is tighter than the Individual sample mutual information (ISMI)~\citep{bu2020tightening}. Therefore, our upper bound in Proposition~\ref{Prop: wassertine based on joint dis} and Corollary~\ref{Cor: compare TV} would also be tighter than the ISMI bound.
\end{remark}

The proposed bound in Proposition~\ref{Prop: total variation} will reduce to the individual sample total variation distance-based bound in~\citep[Corollary~1]{galvez2021tighter}, when the learning algorithm is symmetric. However, we may want to use \emph{non-symmetric} learning algorithm in practice since the importance of each training sample is not the same, e.g.,  imbalanced classification or learning under noisy data samples. As we will show in Section~\ref{sec: numerical}, for a non-symmetric learning algorithm, our proposed upper bound will be strictly tighter than the bound in~\citep[Corollary~1]{galvez2021tighter}.


\subsection{KL Divergence-based Upper Bound}
In the following theorem, we provide an upper bound in terms of KL divergence using~\eqref{Eq: new rep expected GE} under sub-Gaussian condition.
\begin{theorem}\label{Theorem: KL result}
Suppose that the loss function $\ell(w,z)$ is $\sigma$-sub-Gaussian\footnote{A random variable $X$ is $\sigma$-sub-Gaussian if $E[e^{\lambda(X-E[X])}]\leq e^{\frac{\lambda^2 \sigma^2}{2}}$ for all $\lambda \in \mathbb{R}$.} under distribution $ P_W \otimes \overline{P}_{\overline{Z}}$. The following upper bound holds on the expected generalization error
\begin{align}
    \overline{\text{gen}}(P_{W|S},P_S)\leq \sqrt{2\sigma^2 D( \overline{P}_{W,\overline{Z}}\|P_W \otimes \overline{P}_{\overline{Z}})}.
\end{align}
\end{theorem}
\begin{proof}[\textbf{Sketch of Proof}]
Applying the Donsker-Varadhan representation of KL divergence~\eqref{eq: KL rep} to the generalization error expressed in~\eqref{Eq: new  rep expected GE} and using the $\sigma$-sub-Gaussianity in a similar approach to \citep[Lemma~1]{xu2017information}, it completes the proof.
\end{proof}
In the following, we compare our KL divergence based upper bound  with  the  mutual information based bound in \citep[Theorem~1]{xu2017information}.
\begin{corollary}\label{Cor: ind sample mutual}
Under the same assumption as in Theorem~\ref{Theorem: KL result}, and further assume that training samples $S = \{Z_i\}_{i=1}^n$ are  i.i.d., the upper bound in Theorem \ref{Theorem: KL result}
is no worse than the mutual information-based upper bound in \citep[Theorem~1]{xu2017information}, i.e., 
 \begin{align}
    \overline{\text{gen}}(P_{W|S},P_S)
    &\leq \sqrt{2\sigma^2 D( \overline{P}_{W,\overline{Z}}\|P_W \otimes \overline{P}_{\overline{Z}})} \nn \\
    &\leq \sqrt{\frac{2\sigma^2}{n}  I(W;S)}.
\end{align}
\end{corollary}
\begin{proof} 
Under i.i.d assumption, $P_Z=P_Z$. Then, we have
  \begin{align}
    \overline{\text{gen}}(P_{W|S},P_S)&\leq \sqrt{2\sigma^2 D( \overline{P}_{W,\overline{Z}}\|P_W \otimes P_Z)}\\
    &\leq \sqrt{ \frac{2\sigma^2}{n}\sum_{i=1}^n D(P_{W,Z_i}\|P_W \otimes P_Z)}\\
    &= \sqrt{ \frac{2\sigma^2}{n}\sum_{i=1}^n I(W;Z_i)}\\
    &\leq \sqrt{\frac{2\sigma^2}{n}  I(W;S)},
\end{align}
where the second inequality follows from the convexity of KL divergence, and the last inequality is due to the chain rule of mutual information and the i.i.d assumption~\citep[Proposition~2]{bu2020tightening}. 
\end{proof}
\begin{remark}
Under the same assumption as in Theorem~\ref{Theorem: KL result}, our upper bound in Theorem~\ref{Theorem: KL result} will reduce to the ISMI bound proposed in~\citep[Proposition~1]{bu2020tightening}, when the learning algorithm $P_{W|S}$ is symmetric.
\end{remark}

We can also provide the following generalization error upper bound in terms of the reversed KL divergence using the average joint distribution as in \eqref{Eq: new  rep expected GE}.
\begin{proposition}\label{Prop: Lautum information}
Suppose that the loss function $\ell(w,z)$ is $\sigma$-sub-Gaussian under $  \overline{P}_{W,\overline{Z}}$ distribution. Then, the following upper bound holds
\begin{align}
    \overline{\text{gen}}(P_{W|S},P_S)\leq \sqrt{2\sigma^2 D(P_W \otimes \overline{P}_{\overline{Z}} \|  \overline{P}_{W,\overline{Z}})}.
\end{align}
\end{proposition}
Similar to Corollary~\ref{Cor: ind sample mutual}, we have the following result.
\begin{corollary}
Under the same assumption as in Proposition~\ref{Prop: Lautum information},  the upper bound in Proposition \ref{Prop: Lautum information}
is always no worse than the upper bound based on individual sample Lautum Information,
 \begin{align}
    \overline{\text{gen}}(P_{W|S},P_S)&\leq \sqrt{2\sigma^2 D(P_W \otimes P_Z \|  \overline{P}_{W,\overline{Z}})} \nn  \\
    &\leq \sqrt{\frac{2\sigma^2}{n} \sum_{i=1}^n L(W;Z_i)}.
\end{align}
\end{corollary}
\subsection{Jensen-Shannon Divergence Based Upper Bound}
We can also apply the average joint distribution approach to the Jensen-Shannon divergence based upper bound in \citep{aminian2020jensen}.
\begin{theorem}\label{Prop: upper jensen}
Suppose that the loss function $\ell(w,z)$ is $\sigma$-sub-Gaussian under distribution $\frac{P_W \otimes \overline{P}_{\overline{Z}} +  \overline{P}_{W,\overline{Z}}}{2}$. The following upper bound holds on the expected generalization error
\begin{align}
    |\overline{\text{gen}}(P_{W|S},P_S)|\leq 2\sqrt{2\sigma^2 D_{JS}( \overline{P}_{W,\overline{Z}}\|P_W \otimes \overline{P}_{\overline{Z}})}.
\end{align}
\end{theorem}
\begin{proof}[\textbf{Sketch of Proof}]
The theorem can be proved by using the auxiliary distribution technique in \citep{aminian2020jensen} and considering the generalization error representation in terms of average joint distribution in~\eqref{Eq: new  rep expected GE}.
\end{proof}
As discussed in \citep{polyanskiy2014lecture}, Jensen-Shannon is a $f$-divergence and it is a jointly convex function. Thus, we have:
\begin{align}\nn
    |\overline{\text{gen}}(P_{W|S},P_S)|&\leq 2\sqrt{2\sigma^2 D_{JS}( \overline{P}_{W,\overline{Z}}\|P_W \otimes \overline{P}_{\overline{Z}})}\\\label{eq: js per sample}
    &\leq 2\sqrt{\frac{2\sigma^2 }{n} \sum_{i=1}^n D_{JS}(P_{W,Z_i}\|P_W \otimes P_{Z})},
\end{align}
where \eqref{eq: js per sample} is an upper bound based on per sample Jensen-Shannon divergence.
\section{The Difference of Empirical risks}
We now consider a slightly different setting. Suppose one has access two different learning algorithms $A$ and $B$, i.e. $P_{W_A|S}$ and $P_{W_B|S}$. And the goal is to quantify the difference between the empirical risk associated with each of the learning algorithms, i.e.,
\begin{equation}
  \Delta_E(A,B)= \mathbb{E}_{P_{W_A,W_B,S}} [L_E(W_A,S)-L_E(W_B,S)].
\end{equation}
Using the average joint distribution, we can provide an upper bound on the absolute value of the difference between the empirical risks of these algorithms.
\begin{proposition}\label{Prop: diff emp}
Suppose that the loss, $\ell(w,z)$, is $\sigma$-sub-Gaussian under $\overline{P}_{W_B,\overline{Z}}$ distribution. The following upper bound holds on the expected difference between empirical risks of two learning algorithms,
\begin{align}
    \left| \Delta_E(A,B)\right|\leq \sqrt{2\sigma^2 D( \overline{P}_{W_A,\overline{Z}}\| \overline{P}_{W_B,\overline{Z}})}
\end{align}
\end{proposition}
\begin{proof} $\Delta_E(A,B)$ can be written as 
\begin{align}\nn
   \Delta_E(A,B)&= \mathbb{E}_{P_{W_A,W_B,S}} [L_E(W_A,S)-L_E(W_B,S)]\\\label{Eq: diff emp rep}
   &=\mathbb{E}_{\overline{P}_{W_A,\overline{Z}}}[\ell(W,Z)]-\mathbb{E}_{\overline{P}_{W_B,\overline{Z}}}[\ell(W,Z)].
\end{align}
The final result holds by applying Donsker-Varadhan~\eqref{eq: KL rep} to \eqref{Eq: diff emp rep} and using $\sigma$-sub-Gaussian in a similar way as in \citep[Lemma~1]{xu2017information}.
\end{proof}

In a similar way to Proposition~\ref{Prop: diff emp}, we could provide an upper bound on the difference of two empirical risks achieved using a different number of training samples. Let $W'$ denote the output of the learning algorithm trained with $S^\prime_m$, which contains $m$ samples, and $W$ is learned using $S_n$ with $n$ samples.
\begin{corollary}\label{Cor: erm diff samples}
Suppose that the loss function $\ell(w,z)$ is $\sigma$-sub-Gaussian under distribution $\overline{P}_{W,\overline{Z}}$. We have the following upper bound on the expected difference of empirical risks achieved using different number of training samples 
\begin{align*}
 &\left|\mathbb{E}[L_E(W^\prime,S^\prime_m)-L_E(W,S_n)]\right|  \leq \sqrt{2\sigma^2 D(\overline{P}_{W',\overline{Z}'}\| \overline{P}_{W,\overline{Z}})},
\end{align*}
where the expectation is over the distribution $P_{W^\prime,W,S^\prime_m,S_n}$.
\end{corollary}

\section{Numerical Example}\label{sec: numerical}
We illustrate that the proposed bounds can be tighter than existing ones using a simple toy example. The goal of the example is to estimate the mean of a Gaussian random variable $Z \sim \mathcal{N}(\beta,\sigma^2)$ based on two i.i.d. samples $Z_1$ and $Z_2$. We consider the estimate given by $W=tZ_1+(1-t)Z_2$ for $0<t<1$, and adopt the truncated $\ell_2$ loss function $\ell(w,z)=\min((w-z)^2,c^2)$. Since the loss function is bounded within the interval $[0,c^2]$, it is $\frac{c^2}{2}$-sub-Gaussian for all $w$. In the following, we evaluate four generalization error upper bounds based on different information measures: 1)  Individual sample mutual information proposed in~\citep[Proposition~1]{bu2020tightening}, 2) KL divergence using average joint distribution in Theorem~\ref{Theorem: KL result}, 3) individual sample total variation distance in \citep[Corollary 1]{galvez2021tighter}, and 4) total variation using average joint distribution in Proposition~\ref{Prop: total variation}. Thus, we have
\begin{align}\label{Eq: Mutual information Based Bound}
    &\overline{\text{gen}}(P_{W|Z_1,Z_2},P_Z)\leq \frac{c^2}{4}\left(\sqrt{2I(W;Z_1)}+\sqrt{2I(W;Z_2)}\right),\nn \\
   &\overline{\text{gen}}(P_{W|Z_1,Z_2},P_Z)\leq\frac{c^2}{2} \sqrt{2D(\overline{P}_{W,\overline{Z}}\|P_W \otimes P_Z)},\\
   &\overline{\text{gen}}(P_{W|Z_1,Z_2},P_Z)\leq\\\nn&\quad \frac{c^2}{2} \left(\mathbb{TV}(P_{W,Z_1},P_W \otimes P_Z)+\mathbb{TV}(P_{W,Z_2},P_W \otimes P_Z)\right),\\
   &\overline{\text{gen}}(P_{W|Z_1,Z_2},P_Z)\leq c^2 \mathbb{TV}(\overline{P}_{W,\overline{Z}},P_W \otimes P_Z).
\end{align}
It can be shown that $W \sim \mathcal{N}(\beta,\sigma^2(t^2+(1-t)^2))$, and $(W,Z_1)$ and $(W,Z_2)$ are jointly Gaussian with correlation coefficients $\rho_1=\frac{t}{\sqrt{t^2+(1-t)^2}}$ and $\rho_2=\frac{(1-t)}{\sqrt{t^2+(1-t)^2}}$, respectively. Note that
\begin{align}
    D(\overline{P}_{W,\overline{Z}}\|P_W \otimes P_Z)=h(P_W)+h(P_Z)-h(\overline{P}_{W,\overline{Z}}),
\end{align}
with $h (\cdot)$ denoting the differential entropy, i.e., 
\begin{align*}
 &h(P_Z)=\frac{1}{2}\log(2\pi\sigma^2e),\\
  &h(P_W)=\frac{1}{2}\log(2\pi\sigma^2(t^2+(1-t)^2)e),
\end{align*}
 whereas $h(\overline{P}_{w,Z^2})$ can be computed numerically.

\begin{figure}
    \centering
    \includegraphics[scale=0.23]{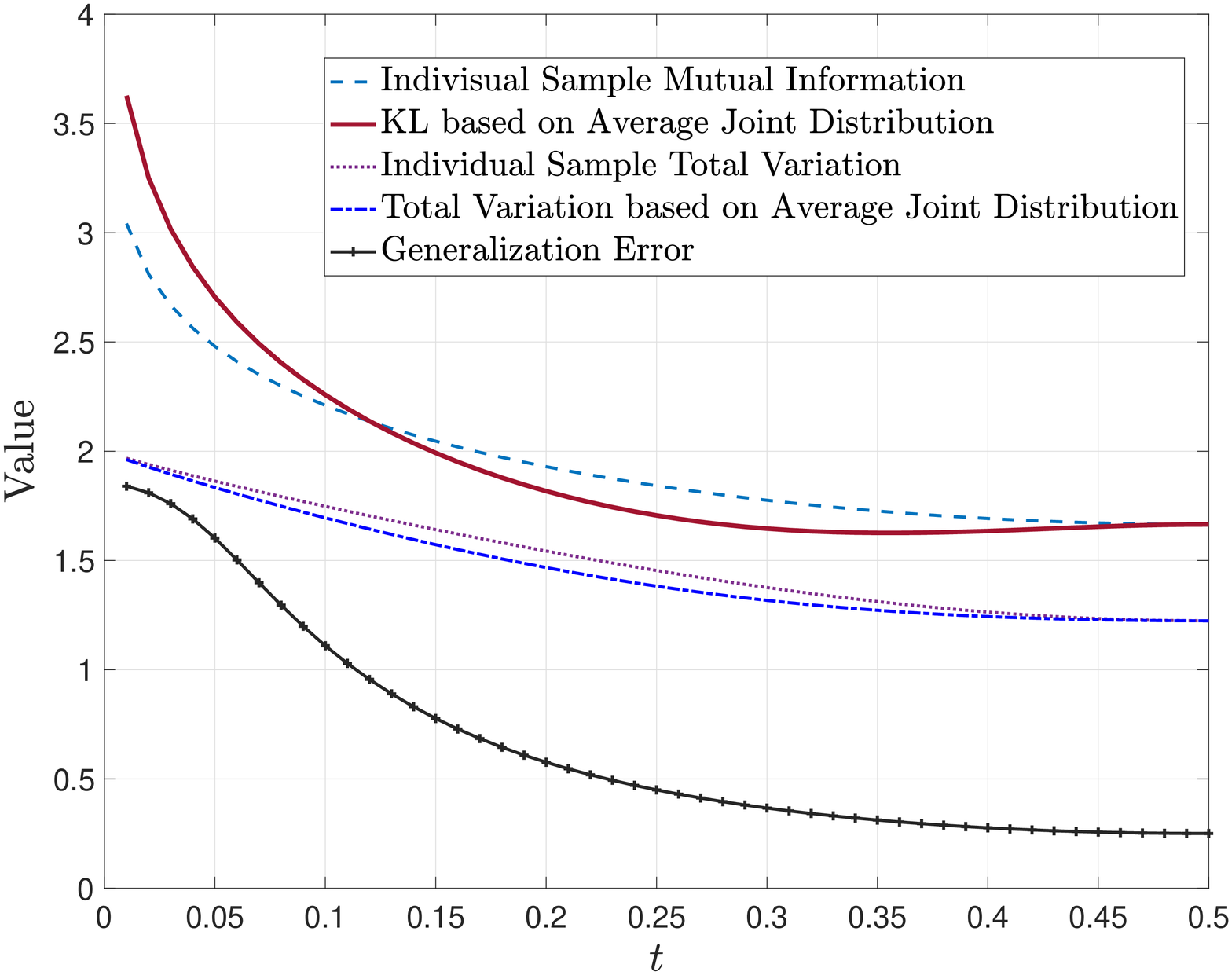}
    \caption{Comparison of the true generalization error and four generalization error upper bounds in Gaussian mean estimation example with $\sigma=10$ and $c =2 $, as we change $t$.}
    \label{fig:bound compare}
\end{figure}
Fig.\ref{fig:bound compare} depicts the four generalization error bounds based on individual sample mutual information, KL divergence using average joint distribution, individual sample total variation distance, total variation distance using average joint distribution, and the true generalization error. It can be seen that for $t>0.1$, the upper bound based on KL divergence using average joint distribution is tighter than the individual sample mutual information-based upper bound. In addition, the total variation using average joint distribution gives the tightest upper bound. At $t=0.5$, the learning algorithm would be symmetric with respect to $Z_1$ and $Z_2$. Therefore, the individual sample mutual information-based upper bound equals KL divergence-based upper bound using average join distribution. Similarly, our total variation distance-based upper bound using average joint distribution is equal to the individual sample total variation distance-based upper bound at $t=0.5$.

\section{Conclusion}
We have introduced a new approach to obtain information-theoretic bounds of the generalization error for supervised learning problems. Our upper bounds based on Wasserstein distance and total variation distance are tighter than counterparts based on individual samples. Our approach could also be combined with PAC-Bayesian upper bounds \citep{van2014pac} and conditional information techniques \citep{steinke2020reasoning} to tighten the result, which is left for future research.
 \bibliographystyle{IEEEtran}
 \bibliography{Refs}

\end{document}

%% file: math_commands.tex
\usepackage{amsmath,amsfonts,bm}

\def\1{\bm{1}}

\DeclareMathAlphabet{\mathsfit}{\encodingdefault}{\sfdefault}{m}{sl}
\SetMathAlphabet{\mathsfit}{bold}{\encodingdefault}{\sfdefault}{bx}{n}

